\documentclass[pre,showpacs,twocolumn,reprint,superscriptaddress]{revtex4-1}
\usepackage{graphicx,amsmath,amssymb,amsthm,braket}
\theoremstyle{definition}
\newtheorem{assumption}{Assumption}
\newtheorem{thm}{Theorem}
\newtheorem{corollary}{Corollary}
\newtheorem*{conjecture}{Conjecture}
\newtheorem{lemma}{Lemma}
\DeclareMathAlphabet{\mathpzc}{OT1}{pzc}{m}{it}
\newcommand{\dd}{\mathrm{d}}
\newcommand{\ee}{\mathrm{e}}
\newcommand{\ii}{\mathrm{i}}
\begin{document}
\title{The Second Law of Thermodynamics under Unitary Evolution and External Operations}
\author{Tatsuhiko N. Ikeda}
\affiliation{Department of Physics, University of Tokyo, Bunkyo-ku, Tokyo 113-0033, Japan}
\affiliation{Physics Department, Boston University, Boston, MA 02215, USA}
\author{Naoyuki Sakumichi}
\affiliation{Theoretical Research Division, Nishina Center, RIKEN, Wako, Saitama 351-0198, Japan}
\author{Anatoli Polkovnikov}
\affiliation{Physics Department, Boston University, Boston, MA 02215, USA}
\author{Masahito Ueda}
\affiliation{Department of Physics, University of Tokyo, Bunkyo-ku, Tokyo 113-0033, Japan}
\date{\today}
\begin{abstract}
A microscopic definition of the thermodynamic entropy in an isolated quantum system must satisfy (i) additivity, (ii) extensivity and (iii) the second law of thermodynamics. We show that the diagonal entropy, which is the Shannon entropy in the energy eigenbasis at each instant of time, meets the first two requirements and that the third requirement is satisfied if an arbitrary external operation is performed at typical times. In terms of the diagonal entropy, thermodynamic irreversibility follows from the facts that the Hamiltonian dynamics restricts quantum trajectories under unitary evolution and that the external operation is performed without referring to any particular information about the microscopic state of the system. 
\end{abstract}
\pacs{05.30.-d, 03.65.-w}
\maketitle
\section{Introduction}.
Entropy is a concept in thermodynamics consistent with all empirical facts on a macroscopic scale, satisfying (i) additivity, (ii) extensivity, and (iii) the second law of thermodynamics. Here (iii) dictates that, upon an external operation on an isolated system,
entropy stays constant if the operation is quasi-static
and otherwise increases~\cite{Lieb1999}.
However, a quantum-mechanical definition of the thermodynamic entropy
which satisfies all the three requirements has not been established~\cite{balian2007microphysics}.
For example,
the von Neumann entropy (vN-entropy)
$S_{\text{vN}}(\hat{\rho})\equiv-\text{tr}(\hat{\rho}\ln\hat{\rho})$,
where $\hat{\rho}$ is the density operator of the system
(the Boltzmann constant is set to unity in this paper),
does not satisfy (iii)
because
the vN-entropy stays constant under any unitary evolution
including non-quasi-static external operations~\cite{landau1996statistical}.
Now that
isolated quantum systems have experimentally been realized, for example, 
in ultracold atomic systems~\cite{Greiner2002,Kinoshita2006,Trotzky2012},
it seems urgent to seek for a microscopic definition of the thermodynamic entropy.

The difficulty in deriving the second law from quantum mechanics
lies in the fact that the former asserts irreversibility while the latter presupposes unitary evolution, which is reversible.
For essentially the same reason, it is also nontrivial how thermalization occurs in an isolated quantum system.

The notion of ``typicality''
has recently advanced 
the understanding for the latter problem~\cite{Polkovnikov2011a,Yukalov2011}.
A great majority of pure states in a microcanonical energy shell,
or typical pure states,
have been shown to represent thermal equilibrium~\cite{Popescu2006,Goldstein2006,Sugita2006,Reimann2007}.
Furthermore,
it has also been shown that,
during
a long-time unitary evolution with a time-independent Hamiltonian $\hat{H}$,
a stationary state is seen at almost every time, or typical times,
under proper assumptions (see Assumptions~1 and 2 below)~\cite{Tasaki1998,Reimann2008a,Rigol2008,Linden2009,
Neumann2010,Goldstein2010a,Goldstein2010b}.
The reason why the stationary state is seen at typical times is
because recurrences occur repeatedly in unitary evolution, which is reversible~\cite{Bocchieri1957}.
However, since the recurrence time grows exponentially with the degrees of freedom of the system,
we rarely, if ever, encounter them.
We note that
all the information about the stationary state 
is encoded only in the diagonal elements of a state $\hat{\rho}$ in the eigenbasis of $\hat{H}$
because
the off-diagonal elements oscillate at different frequencies and are averaged out~\cite{Reimann2008a,Rigol2008}.

The diagonal entropy (d-entropy)~\cite{Polkovnikov2011} describes the thermodynamic entropy of the stationary state
and is defined by
\begin{align}\label{eq:diagonal_entropy}
S(\hat{\rho})\equiv-\sum_n\rho_{nn}\ln\rho_{nn},
\end{align}
where $\rho_{nn}$ are the diagonal elements of $\hat{\rho}$ in the energy eigenbasis at each instant of time.
The d-entropy consists only of the diagonal elements because they retain the information about the stationary state.
In particular, if $\hat{\rho}$ is a diagonal operator such as the canonical ensemble, the d-entropy coincides with the vN-entropy.
We note that the d-entropy has been shown to possess 
(i) additivity
and 
(ii) extensivity~\cite{Polkovnikov2011}.

The d-entropy stays constant during unitary evolution with a time-independent Hamiltonian
because the stationary state does not change.
In this paper we do not discuss
the entropy for non-stationary states,
which is believed to monotonically increase as the system approaches the stationary state;
however, it still remains unclear if such entropy can be defined locally in time~\cite{Lieb2013a}.
Nevertheless, the d-entropy does not contradict Boltzmann's H-theorem
because the sum of d-entropies of subsystems grows in time as the system approaches a stationary state~\cite{Polkovnikov2011}.

In this paper
we show that the d-entropy satisfies (iii)
if an arbitrary unitary operation is performed at typical times.
Considering that a stationary state is achieved at typical times as mentioned above,
it seems sufficient that (iii) holds for typical operation times
because (equilibrium) thermodynamics assigns entropy only to equilibrium states
and the second law refers to an external operation which changes one equilibrium state to another.
Furthermore, no function is known to satisfy (iii) 
for any $\hat{\rho}$ and unitary operation $\hat{V}$.
This is because
because, if $\hat{V}$, which changes $\hat{\rho}$ into $\hat{\rho}'=\hat{V}\hat{\rho}\hat{V}^\dagger$,
increases $\mathcal{S}$, then the inverse operation $\hat{V}^\dagger$
decreases $\mathcal{S}$ when applied to $\hat{\rho}'$.
Thus, we believe that the d-entropy is, at this moment, the best
microscopic definition of the thermodynamic entropy regarding (i), (ii), and (iii).

\section{Setup}.
Let the dimension of the Hilbert space $\mathpzc{H}$ be $D$,
which is either finite or countably infinite.
At time $t=0$ the state of the system is represented by the density matrix
$\hat{\rho}^0$, which is either pure or mixed.
From $t=0$ to $\tau$,
the evolution is governed by a time-independent Hamiltonian $\hat{H}$.
At $t=\tau$ we perform an external operation represented by an arbitrary unitary operator $\hat{V}$~\cite{CommentOperation}.
The state immediately after the operation is given by
$\hat{\rho}'(\tau)=\hat{V}\ee^{-\ii\hat{H}\tau/\hbar}\hat{\rho}^0\ee^{\ii\hat{H}\tau/\hbar}\hat{V}^\dagger$.
After the operation the evolution is governed by a time-independent Hamiltonian $\hat{H}'$.
The operation is either cyclic 
($\hat{H}'=\hat{H}$) or otherwise ($\hat{H}'\neq\hat{H}$).
Let $\{\ket{E_n}\}_{n=1}^D$ and
$\{\ket{E'_n}\}_{n=1}^D$ be the eigenstates of the Hamiltonians
$\hat{H}$ and $\hat{H}'$, respectively.
Namely, it follows that
$\hat{H}\ket{E_n}=E_n ket{E_n}$ and 
$\hat{H}'\ket{E'_n}=E'_n\ket{E'_n}$  $(\forall{n}\in\{1,2,\dots,D\})$.
We make the following assumption on the energy spectra $\{E_n\}_{n=1}^D$ and $\{E_n'\}_{n=1}^D$.

\begin{assumption}\label{spectrum}
(A) The energy spectra $\{E_n\}_{n=1}^D$ and $\{E_n'\}_{n=1}^D$ are not degenerate.
(B) 
If $E_{k}-E_{l}=E_{m}-E_{n}\neq0$,
then  $k=m$ and $l=n$.
The same condition also holds for $\{E_n'\}_{n=1}^D$.
\end{assumption}
Assumption~\ref{spectrum}(B),
which implies
that all spacings between two different eigenenergies are different,
is believed to hold in systems of interacting particles
and known as a necessary condition for a stationary state to be reached
in isolated quantum systems~\cite{Tasaki1998,Reimann2008a,Neumann2010,Short2011}.

To show the second law in this setup,
we denote the d-entropies before the operation as $S_0$
and after it as $S'(\tau)$.
Introducing 
$\rho^0_{mn}\equiv\braket{E_m|\hat{\rho}^0|E_n}$ and
$\rho'_{mn}(\tau)\equiv\braket{E_m'|\hat{\rho}'(\tau)|E_n'}$,
we have
$S_0=-\sum_n\rho^0_{nn}\ln\rho^0_{nn}$ and $S'(\tau)=-\sum_n\rho'_{nn}(\tau)\ln\rho'_{nn}(\tau)$
(the unspecified sum $\sum_n$ denotes $\sum_{n=1}^D$ throughout this paper).
We note $\rho'_{mn}(\tau)=\sum_{k,l}U_{mk}U_{nl}^*\ee^{-\ii(E_k-E_l)\tau/\hbar}\rho^0_{kl}$
with $U_{mn}\equiv\braket{E'_m|\hat{V}|E_n}$ 
denoting the transition amplitude between $\ket{E_n}$ and $\ket{E_m'}$ caused by the operation.

\section{Rigorous lower bounds for the d-entropy after the operation}.
Before discussing the general initial state including pure states,
we note that the second law holds
if the initial state of the system is diagonal in the energy eigenbasis:
\begin{align}\label{eq:diagonal_ensemble}
\hat{\rho}_\text{DE}^0\equiv\sum_n\rho^0_{nn}\ket{E_n}\bra{E_n}.
\end{align}
The initial d-entropy is $S_0$
since
the diagonal elements of $\hat{\rho}^0$ and $\hat{\rho}_\text{DE}$ are equal.
The state immediately after the operation is 
$\hat{\rho}'_\text{DE}\equiv\hat{V}\hat{\rho}^0_\text{DE}\hat{V}^\dagger$.
With
$\mu_n\equiv\braket{E_n'|\hat{\rho}'_\text{DE}|E_n'}=\sum_m|U_{nm}|^2\rho_{mm}^0$,
the d-entropy after the operation is given by
$S'_\text{DE}\equiv-\sum_n\mu_n\ln\mu_n$, which satisfies
\begin{align}\label{eq:Sde}
S_0\le S'_\text{DE}
\end{align}
from Jensen's inequality
and the doubly stochastic property of 
$|U_{mn}|^2$
(\textit{i.e.}, $\sum_n|U_{mn}|^2=\sum_m|U_{mn}|^2=1$)~\cite{Polkovnikov2011}.

Although the above argument
ensures the second law if the initial state is diagonal in the energy eigenbasis~\cite{Lenard1978,Tasaki2000},
it is not applicable to a sequence of operations since
the final state is not always diagonal.
To resolve this problem, we discuss an arbitrary initial state $\hat{\rho}^0$
whose off-diagonal elements are not necessarily zero due to superpositions of the energy eigenstates.
Our first result is that
the off-diagonal elements make a negligibly small contribution,
which is less than 1 in the long-time average of the d-entropy
(see Supplemental Material for the proof).
\begin{thm}\label{LTA}
\begin{align}\label{eq:average_bound}
0\le S'_\text{DE}-\overline{S'(\tau)}<1,
\end{align}
where $\overline{F(\tau)}\equiv\lim_{T\to \infty}T^{-1}\int_0^T\dd\tau\,F(\tau)$ for
any function $F(\tau)$.
The equality holds if $\hat{\rho}^0=\hat{\rho}_\text{DE}$ or $U_{mn}$ is diagonal.
\end{thm}

We note that 
the d-entropy, unlike any other observable, is a non-linear function of the density matrix and therefore
$\overline{S'(\tau)}$ is not equal to the d-entropy of $\overline{\hat{\rho}'(\tau)}=\rho_\text{DE}'$.
Theorem~\ref{LTA} asserts that the difference between them caused by the nonlinearity is less than 1.

As discussed later, Theorem~\ref{LTA} together with \eqref{eq:Sde} leads to the second law for $\overline{S'(\tau)}$
when the entropies are extensive,
and therefore 1 on the right-hand side of \eqref{eq:average_bound} can be ignored.

Our second result is the following theorem, which gives the probability of $\tau$
with which $S'(\tau)$ is larger than a certain value
(see Supplemental Material for the proof).
\begin{thm}\label{Prob}
\begin{align}\label{eq:typical_2nd}
{\rm Prob}\left[S_0-S_0^{2/3}<S'(\tau)\right]\geq1-\frac{2S'_\text{DE}}{\,S_0^{4/3}}-\frac{R}{S_0^{4/3}}  .
\end{align}
where $R\equiv\sum_{m,n}\nu_{mn}\ln\mu_m\ln\mu_n$ and
$\nu_{mn}\equiv\sum_{\substack{k,l\\k\neq{l}}}U_{mk}U_{nl}^*U_{nk}^*U_{ml}|\rho^0_{kl}|^2$
and $\text{Prob}[\cdots]$ is defined for a uniform distribution of $\tau$ over $\tau\ge0$~\cite{Tasaki1998,Reimann2008a}.
\end{thm}
As discussed later,
Theorem~\ref{Prob} ensures the second law
with unit probability in the thermodynamic limit,
where $S_0^{2/3}$ is negligible compared with $S_0$.

\section{Extensivity of the d-entropy and two second laws in lattice systems}.
To derive the second law of thermodynamics,
we consider a lattice system with the number of sites $N$,
which is a large but finite positive integer.
We assume that the total Hilbert space
$\mathpzc{H}$
is the direct product of every Hilbert space at site $i$: $\mathpzc{H}=\otimes_{i=1}^N\mathpzc{H}_i$.
We also assume that each local Hilbert space $\mathpzc{H}_i$ has the same dimension $d(\ge2)$.
We note $D=\dim\mathpzc{H}=d^N$.

The extensivity of the d-entropy is guaranteed by the following assumption on $\hat{\rho}^0$.
\begin{assumption}\label{assumption:Deff}
An effective dimension, $D_\text{eff}$, is exponentially large in the number of sites $N$:
\begin{align}\label{eq:small_Deff}
D_\text{eff}\equiv\left[\sum_n(\rho^0_{nn})^2\right]^{-1}=\ee^{+O(N)}.
\end{align}
\end{assumption}
\noindent Since we can derive $\ln D_\text{eff}\le{S_0}\le{S'_\text{DE}}\le\ln{D}=N\ln{d}$ from the concavity of the d-entropy,
Assumption~\ref{assumption:Deff} guarantees that $S_0$ and $S'_\text{DE}$ are extensive, or $O(N)$. 

Assumption~\ref{assumption:Deff} ensures that
the initial state, which is non-stationary in general,
reaches a stationary state during unitary time evolution~\cite{Reimann2008a,Linden2009,Short2011}.
This is consistent with the fundamental requirement of thermodynamics that
a system reaches thermal equilibrium when left alone for a long period of time.
We note that Assumption~\ref{assumption:Deff} holds
unless an extremely small number of energy eigenstates are involved in $\hat{\rho}^0$.
This is because
$D_\text{eff}$ represents
an effective number of the populated energy eigenstates
and the number of energy eigenstates in a given energy shell is of the order of $\ee^{O(N)}$.

To describe the two second laws, we define $\lesssim$ and $\sim$ respectively as $\le$ and $=$
with sub-leading terms in $N$ ignored.
First, Theorem~\ref{LTA} and \eqref{eq:Sde} lead to $S_0-1\le\overline{S'(\tau)}$.
Thus, we obtain Corollary~\ref{ave_2nd} below.
Second, Theorem~\ref{Prob},
together with $R=\ee^{-O(N)}$ (see Supplemental Information for details),
leads to Corollary~\ref{typ_2nd} below.
\begin{corollary}[The second law for the time-averaged d-entropy]\label{ave_2nd}
\begin{align}\label{eq:ave_2nd}
S_0  \lesssim \overline{S'(\tau)}.
\end{align}
\end{corollary}
\begin{corollary}[The second law for typical operation times]\label{typ_2nd}
\begin{align}\label{eq:typ_2nd}
\text{Prob}\left[S_0\lesssim{S'(\tau)}\right]\sim1.
\end{align}
\end{corollary}

We note that the inequality $S_0\le\overline{S'(\tau)}$, which is stronger than \eqref{eq:ave_2nd}, does not hold in general.
If the initial state is, for example, at the infinite temperature, or $\rho^0_{nn}=1/D$ $(n=1,2,\dots,D)$
and not every $\rho^0_{mn}$ is zero for $m\neq{n}$,
the d-entropy decreases upon an external operation.
However, our results show that the decrease is sub-extensive and thus negligible for almost all $\tau$.
Namely, once the state reaches the infinite temperature, the d-entropy saturates
within sub-extensive fluctuations.

We also note that Corollary~\ref{typ_2nd} does not contradict the inverse process
mentioned at the end of the introduction.
Here we fix $\tau=\tau_1$
and suppose $S_0 < S'(\tau_1)$.
Let $\hat{\rho}'(\tau_1)$ evolves with $\hat{H}'$ for a period of time $\tau_2$
and $\hat{\rho}^f_{\tau_2}=\ee^{-\ii\hat{H}'\tau_2/\hbar}\hat{\rho}'(\tau)\ee^{\ii\hat{H}'\tau_2/\hbar}$
be the final state of the forward process. 
We define the inverse process such that
the initial state is $\hat{\rho}^f_{\tau_2}$ at $t=0$,
evolves according to $-\hat{H}'$ from $t=0$ to $\tilde{\tau}$,
and a unitary external operation $\hat{V}^\dagger$ is performed at $t=\tilde{\tau}$.
The operation $\hat{V}^\dagger$ decreases the d-entropy for $\tilde{\tau}=\tau_2$,
whereas Corollary~\ref{typ_2nd} states that $\hat{V}^\dagger$ increases the d-entropy for almost all $\tilde{\tau}$.
Since the latter statement does not hold for all $\tilde{\tau}$, these two statements are consistent.

This argument
accentuates the need for
choosing properly not only what operation to perform but also when to do it to decrease the d-entropy.
In particular,
unless we know the timing to do it based on
information about the microscopic state of the system,
the increase of the d-entropy is unavoidable.

\section{Geometrical interpretation of the asymmetry between entropy increase and decrease}.
To understand the physical origin of the asymmetry between entropy increase and decrease,
let us consider a qubit, where $D=2$ and two energy eigenstates are $\ket{0}$ and $\ket{1}$.
Each quantum state, which is described by a $2\times2$ density matrix $\hat{\rho}$,
is represented by a point on or inside a unit sphere called the Bloch sphere~\cite{nielsen2000quantum}
(see Fig.~\ref{bloch}a).
Here the $x$-, $y$-, and $z$-coordinates
are given by $\text{tr}[\hat{\rho}\sigma_a]$ $(a=x,y,z)$,
where $\sigma_a$'s are the Pauli matrices.
The d-entropy of each state depends only on its $z$-component
and monotonically increases from the poles to the equatorial plane (see Fig.~\ref{bloch}b).

\begin{figure}
\begin{center}
\includegraphics[width=8.5cm]{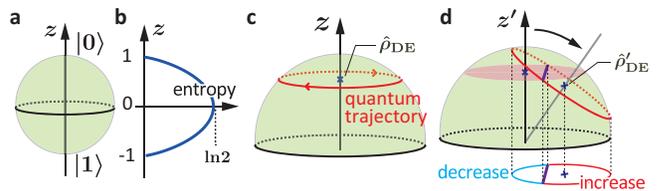}
\caption{(Color Online)
(a) The density matrix is represented by a point on or inside the Bloch sphere.
(b) The d-entropy~\eqref{eq:diagonal_entropy} depends only on the $z$-component and
attains the maximum value of $\ln2$ on the equatorial plane.
(c) A pure state traces a circular quantum trajectory perpendicular to the $z$-axis
under unitary evolution for a fixed Hamiltonian.
The corresponding diagonal ensemble $\hat{\rho}_\text{DE}$ is represented by the center of the trajectory.
(d) A unitary external operation and the associated change
in the energy eigenbasis, which tilts the axis of the trajectory.
The entropy of $\hat{\rho}_\text{DE}$
increases upon the operation
and the entropy increases for more than half of the pure states on the tilted trajectory.
}
\label{bloch}
\end{center}
\end{figure}

During the unitary evolution with the time-independent Hamiltonian,
a pure state traces a circle around the $z$-axis
which we call a quantum trajectory.
The corresponding diagonal ensemble $\hat{\rho}_\text{DE}$
lies at the center of the trajectory (see Fig.~\ref{bloch}c).

A unitary external operation, together with a change in the energy eigenbasis,
causes a unitary transformation on the quantum trajectory,
which is represented by a rotation on the Bloch sphere.
Thus, the quantum trajectory tilts
with $\hat{\rho}_\text{DE}$ changing into $\hat{\rho}_\text{DE}'$ (see Fig.~\ref{bloch}d).
Since $\hat{\rho}_\text{DE}'$ is closer to the equatorial plane than $\hat{\rho}_\text{DE}$, 
the d-entropy is more likely to increase
than decrease for the states on the trajectory (see Fig.~\ref{bloch}d).
Thus, $S_0<S'(\tau)$ holds with more than 50\% probability~\cite{InfiniteTemp}.
We note that Corollary~\ref{typ_2nd} asserts that, as $D$ increases,
this asymmetry is enhanced and the probability that the d-entropy decreases becomes vanishingly small.

The asymmetry in the entropy change
originates from the fact that the quantum trajectory is 
restricted to a circle around the $z$-axis,
which breaks the invariance of the Hilbert space under an arbitrary unitary transformation.
No matter how complex the system might be,
a well-defined Hamiltonian specifies a preferred basis,
thereby restricting possible quantum trajectories.

\section{Universal sub-extensive correction in the d-entropy change}.
Sub-extensive corrections in the d-entropy can be observed
in the systems having $D$ not more than $10^5$.
We conjecture that
$S_\text{DE}'-\overline{S'(\tau)}$ is actually bounded by the following inequality,
which is tighter than \eqref{eq:average_bound}.
\begin{conjecture}\label{conjecture:1-gamma}
\begin{align}
0\le{S'_\text{DE}}-\overline{S'(\tau)}\le1-\gamma,\label{eq:true_bound}
\end{align}
where $\gamma=0.5772\dots$ is Euler's constant.
The equality on the right-hand side of \eqref{eq:true_bound}
holds
no matter what the initial state and the operation are,
as long as
the initial state is pure and the operation causes
numerous transitions between energy eigenstates.
\end{conjecture}

The conjecture is obtained as follows.
We can easily show
$\overline{\rho_{nn}'(\tau)}=\mu_n$ and $\overline{\rho_{nn}'(\tau)^2}\le2\mu_n^2$,
from which we conjecture
$\overline{\rho_{nn}'(\tau)^{1+\epsilon}}\le\Gamma(2+\epsilon)\mu_n^{1+\epsilon}$ $(0\le\epsilon\le1)$
where $\Gamma(x)$ is the gamma function.
This leads to the right inequality in \eqref{eq:true_bound}
in the limit $\epsilon\to0$ since $\Gamma'(2)=1-\gamma$.

We give numerical evidence for the conjecture.
We consider 5 hard-core bosons on a two-dimensional lattice of twenty sites,
which are arranged in a $4\times 5$ rectangle with free boundaries.
The coordinates at site $i$ are denoted by $(x_i,y_i)$ where $x_i=0,1,2,3$ and $y_i=0,1,2$.
The total Hamiltonian is the sum of the kinetic energy
$\hat{H}_\text{kin}=-J\sum_{\langle{i,j}\rangle}(\hat{b}_i^\dagger\hat{b}_j+\hat{b}_j^\dagger\hat{b}_i)$
and the potential energy $\hat{H}_\text{pot}=g\sum_iy_i\hat{n}_i$, 
where
$\sum_{\langle{i,j}\rangle}$ is taken over all the nearest neighbors,
$\hat{b}_i^\dagger$ ($\hat{b}_i$) is the creation (annihilation) operator of the boson on site $i$,
$\hat{n}_i\equiv\hat{b}_i^\dagger\hat{b}_i$ is the number operator of particles on site $i$,
and $g$ is the magnitude of the linear potential along the $y$-axis.
With the hard-core condition, which prohibits more than one boson occupying the same site,
the dimension of the Hilbert space is $D=\binom{20}{5}=15,504$.
The state is initially prepared in the 1000th eigenstate with $g=0$.
At $t=0$, a linear potential $\hat{H}_{\rm{pot}}$ with $g>0$ is switched on.
Then the state becomes nonstationary and approaches a stationary state as time advances.

In Fig.~\ref{fig:sub-extensive}, $S_\text{DE}'-\overline{S'(\tau)}$ is plotted against $g/J$,
where the average $\overline{S'(\tau)}$ is taken over the interval $[500,1000]\hbar/J$
and the error bar is the standard deviation of $S'(\tau)$ in this interval.
\begin{figure}
\begin{center}
\includegraphics[width=8.5cm]{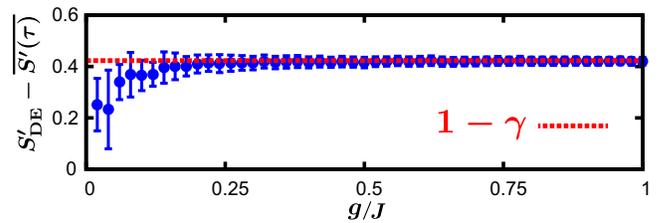}
\caption{(Color Online)
$S_\text{DE}'-\overline{S'(\tau)}$ plotted against the magnitude of the quantum quench, $g/J$,
with error bars showing the standard deviations.
All the data are consistent with \eqref{eq:true_bound}
and $S_\text{DE}'-\overline{S'(\tau)}$ is independent of $g/J$
for $g/J\gtrsim 0.25$.
}
\label{fig:sub-extensive}
\end{center}
\end{figure}

Figure~\ref{fig:sub-extensive} has two implications.
First, it is consistent with \eqref{eq:true_bound}
since all the data points are between 0 and $1-\gamma$.
Second, for $g/J\gtrsim0.25$,
$S_\text{DE}'-\overline{S'(\tau)}$ takes on the value of $1-\gamma$
independently of $g/J$.
We note that this universal constant was also found numerically
in the analysis of the d-entropy in periodically driven systems~\cite{DAlessio2013}.

For a pure state and a unitary external operation which causes numerous transitions between energy eigenstates,
we can analytically derive the universal constant $1-\gamma$
by the replica trick~\cite{Edwards1975}
or from the fact that
each $\rho'_{nn}(\tau)$ obeys an exponential distribution
whose mean is $\mu_n$~\cite{fullpaper}.

The constant $1-\gamma$ can be a signature for the state to be pure
if we can measure the eigenenergies of the system or $D$ diagonal elements of the density matrix.
Since we usually have to know $D^2$ elements of the density matrix including the off-diagonal elements
to measure the purity of the state~\cite{Bagan2005},
the above criterion of judging the purity of the state
is expected to have a practical advantage.

\section{Conclusion}.
In this paper
we have shown that
the d-entropy
satisfies the second law for typical operation times (Corollary~\ref{typ_2nd})
under two assumptions, which ensures that a stationary state is reached
in an isolated quantum system.
This result implies that the d-entropy is, at this moment,
the best microscopic definition of the thermodynamic entropy
regarding additivity, extensivity, and the second law.
In terms of the diagonal entropy,
irreversibility occurs because
the Hamiltonian dynamics restricts the physically allowed quantum states
and the operation is performed without referring to any information about the microscopic state of the system.
We have also shown that the entropy change caused by a unitary external operation includes
the universal sub-extensive correction $1-\gamma$.

\section{Acknowledgements}.
Fruitful discussions with Masahiro Hotta, Takashi Mori, Takahiro Sagawa, and Takanori Sugiyama are gratefully acknowledged.
This work was supported by
KAKENHI 22340114, 
a Grant-in-Aid for Scientific Research on Innovation Areas ``Topological Quantum Phenomena'' (KAKENHI 22103005),
and the Photon Frontier Network Program,
from MEXT of Japan.
T.N.I. acknowledges the JSPS for
financial support (Grant No. 248408).
N.S. was supported by a Grant-in-Aid for JSPS Fellows
(Grant No. 250588).

\section{Note added}.
We became aware of independent work by Tasaki and coworkes~\cite{Tasaki2000a,Goldstein2013}.
While they also discuss the second law in isolated quantum systems with the operation times uniformly distributed,
they only address cyclic operations ($\hat{H}=\hat{H}'$)
in which case only the total energy is needed for discussion.

\bibliographystyle{apsrev4-1}
\bibliography{Preprint_2ndlaw}

\begin{thebibliography}{34}
\bibitem [{Com()}]{CommentOperation}%
  \BibitemOpen
  \href@noop {} {}\bibinfo {note} {When the Hamiltonian depends on a parameter
  $\lambda$, which changes in time as $\lambda=\lambda(t)$ from $t=\tau$ to
  $\tau+\Delta \tau$, $\hat{V} = \mathbb{T} \exp \left[ -\ii
  \int_\tau^{\tau+\Delta \tau} \hat{H}_{\lambda (t)} \dd t \right]$, where
  $\mathbb{T}$ is the time-ordering product.}\BibitemShut {Stop}%
\bibitem [{Inf()}]{InfiniteTemp}%
  \BibitemOpen
  \href@noop {} {}\bibinfo {note} {We note that this argument is valid only if
  the quantum trajectory never crosses the equatorial plane. However, this
  condition, in practice, does not introduce any restriction in the
  thermodynamic limit because the equatorial plane, which corresponds to the
  infinite temperature, cannot be reached.}\BibitemShut {Stop}%
\end{thebibliography}

\appendix
\vspace{1cm}
\begin{center}
\textbf{{\large Appendix}}
\end{center}

This material supplements technical details omitted in the manuscript.
In Section~I, we give two lemmas to prove
the theorems in the manuscript. 
In Sections~\ref{supp:proof:thm1} and \ref{supp:proof:thm2},
we prove Theorem 1 and 2, respectively.
In Section~\ref{supp:smallness_R},
we show that $R$ is exponentially small in the number of sites under Assumption 2.
We set $\hbar=k_B=1$ in this Supplemental Material.

In the following discussion, we assume $\mu_n >0$ without loss of generality
because $n$ such that $\mu_n=0$ does not contribute to the d-entropy as shown below,
where
\begin{align}\label{supp:eq:def:mu}
\mu_n \equiv \sum_{k} |U_{nk}|^2 \rho^0_{kk}.
\end{align}
To justify the assumption $\mu_n >0$, we note
\begin{align}\label{supp:eq:rho0ineq}
|\rho_{kl}^0|^2 \le \rho_{kk}^0 \rho_{ll}^0,
\end{align}
which follows from the fact that the density matrix $\hat{\rho}^0$ is positive-semidefinite
[see, e.g., R.~A.~Horn and C.~R.~Johnson, {\it Matrix Analysis} (Cambridge University Press, 1990), pp.~398].
Now we suppose $\mu_n = 0$.
Then, Eq.~\eqref{supp:eq:def:mu} implies that $U_{nk} =0$ or $\rho^0_{kk}=0$ holds for all $k$.
If $\rho^0_{kk}=0$, \eqref{supp:eq:rho0ineq} tells us that $\rho^0_{kl}=0$ for all $l$.
Since 
\begin{align}\label{supp:eq:def:rho}
\rho'_{nn}(\tau) = \sum_{k,l} U_{nk}U_{nl}^* \ee^{-\ii (E_k - E_l)\tau} \rho^0_{kl},
\end{align}
we obtain $\rho'_{nn}(\tau)=0$ for all $\tau$ and $n$ such that $\mu_n=0$.
Thus, $\rho'_{nn}(\tau)$ such that $\mu_n=0$ does not contribute to the d-entropy after the operation
which is given by
\begin{align}\label{supp:eq:def:entropy}
S'(\tau) \equiv -\sum_n \rho'_{nn}(\tau) \ln \rho'_{nn}(\tau).
\end{align}

\section{Lemmas to prove Theorem~1, 2 and 3}\label{supp:proof:thm1}

To prove Theorems 1 and 2 in the manuscript
and Theorem 3 given in Section~\ref{supp:smallness_R} below,
we begin by showing two useful lemmas.
Lemma~\ref{supp:lem:2ndmom}
gives the long-time averages of $\rho'_{nn}(\tau)$ and
that of its correlation over $\tau$ and
Lemma~\ref{supp:ineq:bound:munu} gives its property.
\begin{lemma}\label{supp:lem:2ndmom}
\begin{align}\label{supp:eq:1stmom}
\overline{\rho'_{nn}(\tau)} = \mu_n;
\end{align}
\begin{align}
\overline{\rho'_{mm}(\tau) \rho'_{nn}(\tau)} =\mu_m \mu_n+ \nu_{mn},\label{supp:eq:2ndmom}
\end{align}
where
\begin{align}
\nu_{mn} \equiv \sum_{\substack{k,l \\ k \neq l}}
 U_{mk}  U^*_{ml}   U_{nk}^* U_{nl} | \rho^0_{kl}|^2. \label{eq:def:nu_nm}
\end{align}
\end{lemma} 
\begin{proof}
The left-hand sides of Eqs.~\eqref{supp:eq:1stmom} and \eqref{supp:eq:2ndmom} are
\begin{align}
\overline{\rho'_{nn}(\tau)}
&=\sum_{k,l}
 U_{nk}  U^*_{nl} \overline{\ee^{-\ii ( E_k -E_l)\tau }} \rho^0_{kl}; \label{supp:eq:17-1}\\
\overline{\rho'_{mm}(\tau)\rho'_{nn}(\tau)}
&=\sum_{i,j,k,l} 
 U_{mi}  U^*_{mj}   U_{nk}^* U_{nl}\notag\\
&\qquad\times \overline{\ee^{-\ii ( E_i -E_j-E_k+E_l)\tau }} \rho^0_{ij}(\rho^0_{kl})^*. \label{supp:eq:17-2}
\end{align}
Assumption~1 in the manuscript leads to
\begin{align}\label{supp:eq:LTA:phase-1}
 \overline{\ee^{-\ii (E_k -E_l)\tau }} &=\delta_{kl};\\
 \overline{\ee^{-\ii (E_i -E_j-E_k+E_l)\tau }}
 &=\delta_{ij}\delta_{kl} + \delta_{ik}\delta_{jl} - \delta_{ij}\delta_{jk}\delta_{kl}.\label{supp:eq:LTA:phase-2}
\end{align}
Substituting Eq~\eqref{supp:eq:LTA:phase-1} in Eq.~\eqref{supp:eq:17-1}, we obtain Eq.~\eqref{supp:eq:1stmom}.
Similarly, with Eq~\eqref{supp:eq:LTA:phase-2} substituted in Eq.~\eqref{supp:eq:17-2},
the first term on the right-hand side of Eq.~\eqref{supp:eq:LTA:phase-2}
gives $\mu_m \mu_n$ and the second and third terms give $\nu_{mn}$.
Thus, we obtain Eq.~\eqref{supp:eq:2ndmom}.
\end{proof}

\begin{lemma}\label{supp:ineq:bound:munu}
\begin{align}
\nu_{mn} \le  \tilde{\nu}_{mn} \le \mu_m \mu_n, \label{supp:ineq:bound:nu}
\end{align}
where
\begin{align}
\tilde{\nu}_{mn} \equiv \sum_{k, l} 
 U_{mk}  U^*_{ml}   U_{nk}^* U_{nl} | \rho^0_{kl}|^2.\label{supp:eq:def:nu_bar}\end{align}
\end{lemma} 

\begin{proof}
First, the inequality $\nu_{mn} \le  \tilde{\nu}_{mn}$ follows from
$\tilde{\nu}_{mn}-\nu_{mn} =\sum_k |U_{mk}|^2 |U_{nk}|^2 (\rho^0_{kk})^2 \ge 0 $.

Second, we prove the inequality $\tilde{\nu}_{mn} \le \mu_m \mu_n$.
We note
\begin{align}\label{supp:eq:19}
\tilde{\nu}_{mn} = \frac{1}{2}\sum_{k,l} 
(U_{mk}  U^*_{ml}   U_{nk}^* U_{nl} + \text{c.c.} )| \rho^0_{kl}|^2 .
\end{align}
Here we also note 
\begin{align}\label{supp:eq:Uineq}
U_{mk}  U^*_{ml}   U_{nk}^* U_{nl} + \text{c.c.} \le  |U_{mk}|^2 |U_{nl}|^2 +|U_{ml}|^2 |U_{nk}|^2,
\end{align}
which follows from $\left| U_{mk} U_{nl} - U_{ml} U_{nk} \right|^2 \ge 0$.
From \eqref{supp:eq:19} and \eqref{supp:eq:Uineq}, we obtain
\begin{align}
\tilde{\nu}_{mn} &\le \frac{1}{2}\sum_{k,l} 
(|U_{mk}|^2 |U_{nl}|^2 +|U_{ml}|^2 |U_{nk}|^2 )| \rho^0_{kl}|^2\\
&\le \sum_{k,l} 
|U_{mk}|^2 |U_{nl}|^2  \rho^0_{kk} \rho^0_{ll}\label{supp:ineq:24}
= \mu_m \mu_n,
\end{align}
where we used \eqref{supp:eq:rho0ineq} to obtain \eqref{supp:ineq:24}.
\end{proof}

\section{Proof of Theorem~1}\label{supp:proof:thm1}
To prove Theorem 1,
we begin with the following lemma.
\begin{lemma}\label{supp:KLdiv:ineq}
Let $D$ be an arbitrary positive integer
and $\{p_n\}$ and $\{q_n\}$ be sets of $D$ real numbers
satisfying $p_n\ge 0$ and $q_n >0$ for $n=1,2,\dots,D$
and $\sum_n p_n =\sum_n q_n =1$.
Then, the following inequalities hold:
\begin{equation}\label{eq:lemma1}
 0 \leq \sum_n  p_n  \ln \frac{p_n}{q_n} \leq \sum_n  \frac{p_n^2}{q_n} -1.
\end{equation}
\end{lemma}
\begin{proof}
We can show
\begin{align}
p_n - q_n \leq p_n \ln \frac{p_n}{q_n} \leq  \frac{p_n^2}{q_n} - p_n, \label{eq:lemma1-sub}
\end{align}
which leads to \eqref{eq:lemma1} when summed over $n$,
as follows.
Since \eqref{eq:lemma1-sub} trivially holds for $p_n=0$,
we assume $p_n>0$.
We recall that $\ln x \leq x-1$ holds for $x>0$.
Substituting $p_n/q_n$ and $q_n/p_n$ for $x$,
we have $1 - \frac{q_n}{p_n} \leq  \ln \frac{p_n}{q_n} \leq  \frac{p_n}{q_n} - 1$,
which leads to \eqref{eq:lemma1-sub}.
\end{proof}

Now we prove Theorem~1.
\begin{proof}[Proof of Theorem~1]
We substitute $\{ \rho'_{nn}(\tau) \}$ and $\{ \mu_n \}$ for $\{p_n \}$ and $\{q_n\}$, respectively, in Lemma~\ref{supp:KLdiv:ineq},
obtaining
\begin{align}\label{supp:eq:before_time_average}
0\le -\sum_n \rho'_{nn}(\tau) \ln \mu_n - S'(\tau) \le \sum_n \frac{\rho'_{nn}(\tau)^2}{\mu_n} -1,
\end{align}
where we used Eq.~\eqref{supp:eq:def:entropy}.
By taking the long-time averages of each side of \eqref{supp:eq:before_time_average}
and using Eq.~\eqref{supp:eq:1stmom},
we have
\begin{align}\label{supp:eq:ineq1}
0 \le S'_\text{DE} - \overline{S'(\tau)} \le \sum_n \frac{\overline{\rho'_{nn}(\tau)^2}}{\mu_n} -1,
\end{align}
where
$S'_\text{DE}  \equiv -\sum_n \mu_n \ln \mu_n$.
The inequality $\overline{\rho'_{nn}(\tau)^2} \leq 2\mu_n^2$ follows from 
Lemmas~\ref{supp:lem:2ndmom} and~\ref{supp:ineq:bound:munu}.
From \eqref{supp:eq:ineq1} and $\sum_n \mu_n =1$, we obtain
\begin{align}
0 \le S'_\text{DE} - \overline{S'(\tau)} \le 1.\label{supp:eq:theorem1}
\end{align}

Finally, we discuss the two equalities in \eqref{supp:eq:theorem1}.
The equality on the right-hand side of \eqref{supp:eq:theorem1}
does not hold.
If the equality holds, it follows that $\nu_{nn}= \mu_n^2$ for any $n$.
This leads to $\sum_k |U_{nk}|^4 (\rho^0_{kk})^2=0$ and hence $\mu_n=0$ for any $n$,
which contradicts $\mu_n>0$.
The equality for the inequality on the left-hand side of Eq~\eqref{supp:eq:theorem1}
holds, for example, if the initial state $\hat{\rho}^0$ is diagonal in the eigenbasis of $\hat{H}$
or if the transition matrix $U_{mn}$ is diagonal.
\end{proof}

We mention that the term in the middle of \eqref{eq:lemma1} is known as
the relative entropy or the Kullback-Leibler divergence:
$D_{\text{KL}}( \{p_n\} || \{q_n\} ) \equiv \sum_n p_n  \ln (p_n/q_n)$.  
Then, the following relation holds:
\begin{align}
\overline{D_{\text{KL}}(\{ \rho'_{nn}(\tau)\}|| \{ \mu_n \})} = S'_\text{DE} - \overline{S'(\tau)}.
\end{align}


\section{Proof of Theorem~2}\label{supp:proof:thm2}
We first show the following lemma.
\begin{lemma}\label{supp:lemma:deviation}
\begin{align}
\overline{\left[ S'(\tau) - S_\text{DE}' \right]^2} \le 2S_\text{DE}' +R,\label{supp:eq:deviation}
\end{align}
where 
\begin{align}
R\equiv \sum_{m,n} \nu_{mn} \ln \mu_m \ln \mu_n.\label{supp:eq:Rdef}
\end{align}
\end{lemma}
\begin{proof}
From the inequality on the left-hand side in \eqref{supp:eq:before_time_average}, we have
$0 \le S'(\tau) \le -\sum_n \rho'_{nn}(\tau) \ln \mu_n$.
Using Lemmas~\ref{supp:lem:2ndmom} and~\ref{supp:ineq:bound:munu}, we have
\begin{align}
\overline{S'(\tau)^2} &\le \sum_{m,n} \overline{\rho'_{mm}(\tau) \rho'_{nn}(\tau)} \ln \mu_m \ln \mu_n
\le (S'_\text{DE})^2 +R. \label{supp:eq:30}
\end{align}
Finally, using \eqref{supp:eq:theorem1} and \eqref{supp:eq:30}, we obtain
\begin{align}
\overline{\left[ S'(\tau) - S_\text{DE}' \right]^2}
&= \overline{S'(\tau)^2}-2\overline{S'(\tau)} S_\text{DE}' + (S_\text{DE}')^2\\
&\le 2S_\text{DE}' + R. 
\end{align}
\end{proof}

We note the following lemma, which is similar to Chebyshev's inequality.
\begin{lemma}\label{supp:lemma:chebyshev}
\begin{align}
{\rm Prob} \left[    \left| S'_\text{DE} - S' (\tau)  \right| \geq C     \right]
    \leq  \frac{\overline{\left[ S'(\tau) - S_\text{DE}' \right]^2}}{C^2}, \label{supp:eq:chebyshev}
\end{align}
where $C$ is an arbitrary positive real number. 
\end{lemma}
\begin{proof}
We define
\begin{equation}
  \theta(\tau) \equiv
  \begin{cases}
  1  \quad {\rm for} \,\, \left| S'_\text{DE} - S' (\tau)  \right| \geq C,  \\
  0  \quad {\rm for} \,\, \left| S'_\text{DE} - S' (\tau)  \right| < C,
\end{cases}
\end{equation}
which satisfies
\begin{equation}
    \theta(\tau) \leq  \frac{\left[S'_\text{DE} - S' (\tau)  \right]^2}{C^2}.\label{supp:eq:36}
\end{equation}
Taking the long-time averages of each side of \eqref{supp:eq:36},
we obtain \eqref{supp:eq:chebyshev}.
\end{proof}

\begin{proof}[Proof of Theorem~2.]
From Lemma~\ref{supp:lemma:deviation} and  Lemma~\ref{supp:lemma:chebyshev} with $C=S_0^{2/3}$,
we have
\begin{align}
{\rm Prob} \left[    \left| S'_\text{DE} - S' (\tau)  \right| \geq S_0^{2/3}     \right]
&\le \frac{2S'_\text{DE}}{S_0^{4/3}}+\frac{R}{S_0^{4/3}}.\label{supp:eq:38}
\end{align}
Then, we obtain
\begin{align}
1-\frac{2S'_\text{DE}}{S_0^{4/3}}-\frac{R}{S_0^{4/3}}
& \le {\rm Prob} \left[    \left| S'_\text{DE} - S' (\tau)  \right| < S_0^{2/3}     \right] \\
&\le {\rm Prob} \left[     S'_\text{DE} -S_0^{2/3} < S' (\tau)   \right] \\
&\le {\rm Prob} \left[     S_0 -S_0^{2/3} < S' (\tau)   \right], 
\end{align}
where we used $S_0 \le S'_\text{DE}$ to obtain the last inequality.
\end{proof}

\section{the exponential smallness of $R$}\label{supp:smallness_R}
In this section, we show that $R$ is exponentially small in the number of sites of a lattice system, $N$,
under Assumption 2, which states that an effective dimension is exponentially large in $N$:
\begin{align}
D_\text{eff} \equiv \left[ \sum_n (\rho^0_{nn})^2 \right]^{-1}= \ee^{O(N)}.\label{supp:eq:assumption_Deff}
\end{align}
This result is due to 
the following theorem.
\begin{thm}\label{supp:thm:bound_R}
\begin{align}\label{bound:R:result}
R \le \frac{5 (\ln D)^2}{D_\text{eff}^{1/2}} +\frac{8(\ln D)^2}{D} +\frac{4(\ln D)^2}{D^2}
\end{align}
holds for any integer $D \, (\ge 2)$.
\end{thm}

Theorem~\ref{supp:thm:bound_R} ensures that $R$ is exponentially small in $N$
from Eq.~\eqref{supp:eq:assumption_Deff} and $D=d^N$.

\begin{proof}
Since Lemma~\ref{supp:ineq:bound:munu} leads to
\begin{align}
R \le \tilde{R} \equiv \sum_{m,n} \tilde{\nu}_{mn} \ln \mu_m \ln \mu_n,\label{supp:eq:R_Rbar}
\end{align}
it is sufficient to prove Theorem~\ref{supp:thm:bound_R} for $\tilde{R}$ instead of $R$.
We separate $\tilde{R}$ into two parts as
\begin{align}
\tilde{R}
= \sum_{(m,n)\in A\setminus B}\tilde{\nu}_{mn}\ln \mu_m \ln \mu_n + \sum_{(m,n)\in B}\tilde{\nu}_{mn}\ln \mu_m \ln \mu_n, \label{supp:eq:separationR}
\end{align}
where the sets of indices, $A$ and $B$, are defined as
\begin{align}
A &\equiv \{ (m,n) \in \{1,\dots,D\} \times \{1,\dots,D\} \};\\
B &\equiv \{ (m,n)\in A \ | \  D_\text{eff}^{-1/2}\mu_m \mu_n \le \tilde{\nu}_{mn} \le \mu_m \mu_n \}.
\end{align}
Since Lemma~\ref{supp:ineq:bound:munu} leads to
$\tilde{\nu}_{mn} \le D_\text{eff}^{-1/2}\mu_m \mu_n$ for $(m,n)\in A\setminus B$,
the first term on the right-hand side of Eq.~\eqref{supp:eq:separationR}
is bounded from above
\begin{align}
\sum_{(m,n)\in A\setminus B}\tilde{\nu}_{mn}\ln \mu_m \ln \mu_n
\le D_\text{eff}^{-1/2}S_\text{DE}'^2
\le D_\text{eff}^{-1/2}(\ln D)^2.
\end{align}
Here we used $S_\text{DE}' \le \ln D$ in deriving the last inequality.
With Lemma~\ref{supp:ineq:bound:munu},
the second term on the right-hand side of Eq.~\eqref{supp:eq:separationR}
is bounded from above
\begin{align}
\sum_{(m,n)\in B}\tilde{\nu}_{mn}\ln \mu_m \ln \mu_n
 \le \sum_{(m,n)\in B}\mu_m \mu_n\ln \mu_m \ln \mu_n.
\end{align}
Defining three subsets of $B$ as
\begin{align}
B_1 &= \{ (m,n)\in B \ | \ \mu_m \ge D^{-2} \ \text{and} \ \mu_n \ge D^{-2} \},\\
B_2 &= \{ (m,n)\in B \ | \ \mu_m \ge D^{-2} \ \text{and} \ \mu_n < D^{-2} \},\\
B_3 &= \{ (m,n)\in B \ | \ \mu_m < D^{-2} \ \text{and} \ \mu_n < D^{-2} \},
\end{align}
we accordingly separate $\sum_{(m,n)\in B}\mu_m \mu_n\ln \mu_m \ln \mu_n$ into three parts $\tilde{R}_1 + 2 \tilde{R}_2 + \tilde{R}_3$, where
\begin{align}
\tilde{R}_\alpha \equiv \sum_{(m,n)\in B_\alpha}\mu_m \mu_n\ln \mu_m \ln \mu_n \quad (\alpha=1,2,3).
\end{align}
Then,
\begin{align}
\tilde{R} \le D_\text{eff}^{-1/2}(\ln D)^2 + \tilde{R}_1 + 2 \tilde{R}_2 + \tilde{R}_3. \label{supp:eq:separationRB}
\end{align}

We give upper bounds for each $\tilde{R}_\alpha$ $(\alpha=1,2,3)$.

First, since $\ln \mu_m \ln \mu_n \le 4(\ln D)^2$ for $(m,n)\in B_1$,
we have
\begin{align}
\tilde{R}_1 \le 4(\ln D)^2 \sum_{(m,n)\in B_1}\mu_m \mu_n
\le 4D_\text{eff}^{-1/2} (\ln D)^2,\label{supp:bound:B1}
\end{align}
where the last inequality holds if the following inequality is satisfied:
\begin{align}\label{supp:bound:weight_in_B}
\sum_{(m,n)\in B} \mu_m \mu_n \le D_\text{eff}^{-1/2}.
\end{align}
To show this,
we sum $D_\text{eff}^{-1/2} \mu_m \mu_n \le \tilde{\nu}_{mn}$
over $m$ and $n$ such that $(m,n)\in B$, obtaining
\begin{align}\label{supp:eq:51}
D_\text{eff}^{-1/2} \sum_{(m,n)\in B} \mu_m \mu_n
\le \sum_{(m,n)\in B} \tilde{\nu}_{mn}
\le  \sum_{m,n} \tilde{\nu}_{mn}=D_\text{eff}^{-1},
\end{align}
where the last equality follows from $\sum_m U_{mk}U_{ml}^* = \delta_{kl}$.
Multiplying \eqref{supp:eq:51} by $D^{1/2}_\text{eft}$,
we obtain \eqref{supp:bound:weight_in_B}.

Second, noting that $-\mu_n \ln \mu_n \le 2 D^{-2} \ln D$ since $\mu_n < D^{-2} \leq 4$,
we have
\begin{align}
\tilde{R}_2 &\le \frac{4(\ln D)^2}{D^2} \sum_{(m,n)\in B_2} \mu_m 
\le \frac{4(\ln D)^2}{D}.\label{supp:bound:B2}
\end{align}
Here, the last inequality holds because $\sum_{(m,n)\in A} \mu_m =D$.

Finally, as we did for $\tilde{R}_2$,
we have
\begin{align}
\tilde{R}_3 &\le \frac{4(\ln D)^2}{D^4} \!\! \sum_{(m,n)\in B_3} \!\! 1 
\le  \frac{4(\ln D)^2}{D^2}.\label{supp:bound:B3}
\end{align}
Here, the last inequality follows from $\sum_{(m,n)\in A} 1 =D^2$.

Combining \eqref{supp:eq:R_Rbar} \eqref{supp:eq:separationRB}, \eqref{supp:bound:B1}, \eqref{supp:bound:B2}, and \eqref{supp:bound:B3},
we obtain \eqref{bound:R:result}.
\end{proof}

\end{document}